\documentclass{article}

\usepackage{times}
\usepackage{amssymb,amsmath,amsthm}
\usepackage{mathrsfs}
\usepackage{epsfig}
\usepackage{color}

\newcommand\ie{{\em i.e.}~}

\newcommand\etc{{\em etc.}}

\def\B{\mathscr B}

\def\C{\mathbb C}
\def\D{\mathscr D}

\def\F{\mathscr F}
\def\G{\mathcal G}
\def\H{\mathcal H}

\def\N{\mathbb N}

\def\R{\mathbb R}
\def\S{\mathscr S}

\def\v{\varphi}

\def\0{\boldsymbol 0}

\def\12{{\textstyle\frac12}}
\def\<{\left\langle}
\def\>{\right\rangle}
\def\({\left(}
\def\){\right)}
\def\[{\left[}
\def\]{\right]}
\def\dom{\mathcal D}
\def\lone{\mathsf{L}^{\:\!\!1}}

\def\ltwo{\mathsf{L}^{\:\!\!2}}

\def\e{\mathop{\mathrm{e}}\nolimits}
\def\d{\mathrm{d}}

\def\slim{\mathop{\hbox{\rm s-}\lim}\nolimits}

\newtheorem{Theorem}{Theorem}[section]
\newtheorem{Remark}[Theorem]{Remark}
\newtheorem{Lemma}[Theorem]{Lemma}
\newtheorem{Assumption}[Theorem]{Assumption}

\newtheorem{Proposition}[Theorem]{Proposition}\newtheorem{Definition}[Theorem]{Definition}

\begin{document}


\title{Time delay is a common feature of quantum scattering theory}

\author{S. Richard$^1$\footnote{On leave from Universit\'e de Lyon; Universit\'e
Lyon 1; CNRS, UMR5208, Institut Camille Jordan, 43 blvd du 11 novembre 1918, F-69622
Villeurbanne-Cedex, France.}\hspace{4pt}\footnote{Supported by the Swiss National
Science Foundation.}~~and R. Tiedra de Aldecoa$^2$\footnote{Supported by the Fondecyt
Grant 1090008 and by the Iniciativa Cientifica Milenio ICM P07-027-F ``Mathematical
Theory of Quantum and Classical Magnetic Systems".}}

\date{\small}
\maketitle \vspace{-1cm}

\begin{quote}
\emph{
\begin{itemize}
\item[$^1$] Department of Pure Mathematics and Mathematical Statistics, Centre for
Mathematical Sciences, University of Cambridge, Cambridge, CB3 0WB, United Kingdom
\item[$^2$] Facultad de Matem\'aticas, Pontificia Universidad Cat\'olica de Chile,\\
Av. Vicu\~na Mackenna 4860, Santiago, Chile
\item[] \emph{E-mails:} sr510@cam.ac.uk, rtiedra@mat.puc.cl
\end{itemize}
}
\end{quote}


\begin{abstract}
We prove that the existence of time delay defined in terms of sojourn times, as well
as its identity with Eisenbud-Wigner time delay, is a common feature of two Hilbert
space quantum scattering theory. All statements are model-independent.
\end{abstract}



\section{Introduction}\label{Intro}
\setcounter{equation}{0}

In quantum scattering theory, there are only few results that are completely
model-independent. The simplest one is certainly that the strong limit
$\slim_{t\to\pm\infty}K\e^{-itH}P_{\rm ac}(H)$ vanishes whenever $H$ is a self-adjoint
operator in a Hilbert space $\H$, $P_{\rm ac}(H)$ the projection onto the subspace of
absolute continuity of $H$ and $K$ a compact operator in $\H$. Another famous result
of this type is RAGE Theorem which establishes propagation estimates for the elements
in the continuous subspace of $\H$. At the same level of abstraction, one could also
mention the role of $H$-smooth operators $B$ which lead to estimates of the form
$\int_\R\d t\,\|B\e^{-itH}\varphi\|^2<\infty$ for $\varphi\in\H$.

Our aim in this paper is to add a new general result to this list. Originally, this
result was presented as the existence of global time delay defined in terms of
sojourn times and its identity with Eisenbud-Wigner time delay \cite{Eis48,Wig55}.
This identity was proved in different settings by various authors (see
\cite{AC87,AJ07,AS06,BGG83,dCN02,GT07,GS80,JSM72,Jen81,Mar76,Mar81,MSA92,Rob94,RW89,
Tie06,Tie09,Wan88} and references therein), but a general and abstract statement has
never been proposed. Furthermore, it had not been realised until very recently that
its proof mainly relies on a general formula relating localisation operators to time
operators \cite{RT10}. Using this formula, we shall prove here that the existence and
the identity of the two time delays is in fact a common feature of quantum scattering
theory. On the way we shall need to consider a symmetrization procedure
\cite{AJ07,BO79,GT07,Mar75,Mar81,Smi60,Tie06,Tie08,Tie09} which broadly extends the
applicability of the theory but which also has the drawback of reducing the physical
interpretation of the result.

Quantum scattering theory is mainly a theory of comparison: One fundamental question
is whether, given a self-adjoint operator $H$ in a Hilbert space $\H$, one can find
a triple $(\H_0,H_0,J)$, with $H_0$ a self-adjoint operator in an auxiliary Hilbert
space $\H_0$ and $J$ a bounded operator from $\H_0$ to $\H$, such that the following
strong limits exist
\begin{equation*}
W_\pm:=\slim_{t\to\pm\infty}\e^{itH}J\e^{-itH_0}P_{\rm ac}(H_0)~?
\end{equation*}
Assuming that the operator $H_0$ is simpler than $H$, the study of the wave operators
$W_\pm$ lead to valuable information on the spectral decomposition of $H$.
Furthermore, if the ranges of both operators $W_\pm$ are equal to $P_{\rm ac}(H)\H$,
then the study of the scattering operator $S:=W_+^*W_-$ leads to further results on
the scattering process. We recall that since $S$ commutes with $H_0$, $S$ decomposes
into a family $\{S(\lambda)\}_{\lambda \in \sigma(H_0)}$ in the spectral
representation $\int_{\sigma(H_0)}^\oplus\d\lambda\,\H(\lambda)$ of $H_0$, with
$S(\lambda)$ a unitary operator in $\H(\lambda)$ for almost every $\lambda$ in the
spectrum $\sigma(H_0)$ of $H_0$.

An important additional ingredient when dealing with time delay is a family of
position-type operators which permits to define sojourn times, namely, a family of
mutually commuting self-adjoint operators $\Phi\equiv(\Phi_1,\ldots,\Phi_d)$ in
$\H_0$ satisfying two appropriate commutation assumptions with respect to $H_0$.
Roughly speaking, the first one requires that for some $z\in\C\setminus\R$ the map
$$
\R^d\ni x\mapsto\e^{-ix\cdot\Phi}(H_0-z)^{-1}\e^{ix\cdot\Phi}\in\B(\H_0)
$$
is three times strongly differentiable. The second one requires that all the
operators $\e^{-ix\cdot\Phi}H_0\e^{ix\cdot\Phi}$, $x\in\R^d$, mutually commute. Let
also $f$ be any non-negative Schwartz function on $\R^d$ with $f=1$ in a
neighbourhood of $0$ and $f(-x)= f(x)$ for each $x\in\R^d$. Then, to define the time
delay in terms of sojourn times one has to consider for any $r>0$ the expectation
values of the localisation operator $f(\Phi/r)$ on the freely evolving state
$\e^{-itH_0}\varphi$ as well as on the corresponding fully evolving state
$\e^{-itH}W_-\varphi$. However one immediately faces the problem that the evolution
group $\{\e^{-itH}\}_{t\in\R}$ acts in $\H$ whereas $f(\Phi/r)$ is an operator in
$\H_0$. As explained in Section \ref{SecSymDelay}, a general solution for this
problem consists in introducing a family $L(t)$ of (identification) operators from
$\H$ to $\H_0$ which satisfies some natural requirements (note that in many examples,
one can simply take $L(t)=J^*$ for all $t\in\R$). The sojourn time for the evolution
group $\{\e^{-itH}\}_{t\in\R}$ is then obtained by considering the expectation value
of $f(\Phi/r)$ on the state $L(t)\e^{-itH}W_-\varphi$. An additional sojourn time
naturally appears in this general two Hilbert space setting: the time spent by the
scattering state $\e^{-itH}W_-\varphi$ inside the time-dependent subset
$\big(1-L(t)^*L(t)\big)\H$ of $\H$. Apparently, this sojourn time has never been
discussed before in the literature. Finally, the total time delay is defined for
fixed $r$ as the integral over the time $t$ of the expectations values involving the
fully evolving state $L(t)\e^{-itH}W_-\varphi$ minus the symmetrized sum of the
expectations values involving the freely evolving state $\e^{-itH_0}\varphi$ (see
Equation \eqref{symsym} for a precise definition). Our main result, properly stated
in Theorem \ref{sym_case}, is the existence of the limit as $r\to\infty$ of the total
time delay and its identity with the Eisenbud-Wigner time delay (see \eqref{eqintro1}
below) which we now define in this abstract setting.

Under the mentioned assumptions on $\Phi$ and $H_0$ it is shown in \cite{RT10} how a
time operator for $H_0$ can be defined: With the Schwartz function $f$ introduced
above, one  defines a new function $R_f\in C^\infty\big(\R^d\setminus\{0\}\big)$ and
express the time operator in the (oversimplified) form
$$
T_f:=-\12\big(\Phi\cdot R_f'(H_0')+R_f'(H_0')\cdot\Phi\big),
$$
with $R_f':=\nabla R_f$ and $H_0':=\big(i[H_0,\Phi_1],\ldots,i[H_0,\Phi_d]\big)$ (see
Section \ref{SecIntegral} for details). In suitable situations and in an appropriate
sense, the operator $T_f$ acts as $i\frac\d{\d\lambda}$ in the spectral
representation of $H_0$ (for instance, when $H_0=-\Delta$ in $\ltwo(\R^d)$, this is
verified with $\Phi$ the usual family of position operators, see \cite[Sec.~7]{RT10}
for details and other examples). Accordingly, it is natural to define in this
abstract framework the Eisenbud-Wigner time delay as the expectation value
\begin{equation}\label{eqintro1}
-\big\langle\varphi,S^*[T_f,S]\varphi\big\rangle
\end{equation}
for suitable $\varphi\in\H_0$.

The interest of the equality between both definitions of time delay  is threefold. It
generalises and unifies various results on time delay scattered in the literature. It
provides a precise recipe for future investigations on the subject (for instance, for
new models in two Hilbert space scattering). And finally, it establishes a relation
between the two formulations of scattering theory: Eisenbud-Wigner time delay is a
product of the stationary formulation while expressions involving sojourn times are
defined using the time dependent formulation. An equality relating these two
formulations is always welcome.

In the last section (Section \ref{Sec_Usual}), we present a sufficient condition for
the equality of the symmetrized time delay with the original (unsymmetrized) time
delay. The physical interpretation of the latter was, a couple of decades ago, the
motivation for the introduction of these concepts.

As a final remark, let us add a comment about the applicability of our abstract
result. As already mentioned, most of the existing proofs, if not all, of the
existence and the identity of both time delays can be recast in our framework.
Furthermore, we are currently working on various new classes of scattering systems
for which our approach leads to new results. Among other, we mention the case of
scattering theory on manifolds which has recently attracted a lot of attention. Our
framework is also general enough for a rigorous approach of time delay in the
$N$-body problem (see \cite{BO79,Mar81,OB75,Smi60} for earlier attempts in this
direction). However, the verification of our abstract conditions for any non trivial
model always require some careful analysis, in particular for the mapping properties
of the scattering operator. As a consequence, we prefer to refer to
\cite{AJ07,GT07,Tie06,Tie08,Tie09} for various incarnations of our approach and to
present in this paper only the abstract framework for the time delay.

\section{Operators $\boldsymbol{H_0}$ and $\boldsymbol\Phi$}\label{SecWave}
\setcounter{equation}{0}

In this section, we recall the framework of \cite{RT10} on a self-adjoint operator
$H_0$ in a Hilbert space $\H_0$ and its relation with an abstract family
$\Phi\equiv(\Phi_1,\ldots,\Phi_d)$ of mutually commuting self-adjoint operators in
$\H_0$  (we use the term ``commute'' for operators commuting in the sense of
\cite[Sec.~VIII.5]{RSI}). In comparison with the notations of \cite{RT10}, we add an
index $0$ to all the quantities like the operators, the spaces, \etc

In order to express the regularity of $H_0$ with respect to $\Phi$, we recall from \cite{ABG} that a
self-adjoint operator $T$  with domain $\dom(T)\subset\H_0$ is said to be of class
$C^1(\Phi)$ if there exists $\omega\in\C\setminus\sigma(T)$ such that the map
\begin{equation}\label{Cm}
\R^d\ni x\mapsto\e^{-ix\cdot\Phi}(T-\omega)^{-1}\e^{ix\cdot\Phi}\in\B(\H_0)
\end{equation}
is strongly of class $C^1$ in $\H_0$. In such a case and for each
$j\in\{1,\ldots,d\}$, the set $\dom(T)\cap\dom(\Phi_j)$ is a core for $T$ and the
quadratic form
$
\dom(T)\cap\dom(\Phi_j)\ni\v\mapsto
\langle T\v,\Phi_j\v\rangle-\langle\Phi_j\v,T\v\rangle
$
is continuous in the topology of $\dom(T)$. This form extends then uniquely to a
continuous quadratic form $[T,\Phi_j]$ on $\dom(T)$, which can be identified with a
continuous operator from $\dom(T)$ to its dual $\dom(T)^*$. Finally, the following
equality holds:
$$
\big[\Phi_j,(T-\omega)^{-1}\big]=(T-\omega)^{-1}[T,\Phi_j](T-\omega)^{-1}.
$$
In the sequel, we shall say that $i[T,\Phi_j]$ is essentially self-adjoint on
$\dom(T)$ if $[T,\Phi_j]\dom(T)\subset\H_0$ and if $i[T,\Phi_j]$ is essentially
self-adjoint on $\dom(T)$ in the usual sense.

Our first main assumption concerns the regularity of $H_0$ with respect to $\Phi$.

\begin{Assumption}\label{chirimoya}
The operator $H_0$ is of class $C^1(\Phi)$, and for each $j\in\{1,\ldots,d\}$,
$i[H_0,\Phi_j]$ is essentially self-adjoint on $\dom(H_0)$, with its self-adjoint
extension denoted by $\partial_jH_0$. The operator $\partial_jH_0$ is of class
$C^1(\Phi)$, and for each $k\in\{1,\ldots,d\}$, $i[\partial_jH_0,\Phi_k]$ is
essentially self-adjoint on $\dom(\partial_jH_0)$, with its self-adjoint extension
denoted by $\partial_{jk}H_0$. The operator $\partial_{jk}H_0$ is of class
$C^1(\Phi)$, and for each $\ell\in\{1,\ldots,d\}$, $i[\partial_{jk}H_0,\Phi_\ell]$ is
essentially self-adjoint on $\dom(\partial_{jk}H_0)$, with its self-adjoint extension
denoted by $\partial_{jk\ell}H_0$.
\end{Assumption}

As shown in \cite[Sec.~2]{RT10}, this assumption implies the invariance of
$\dom(H_0)$ under the action of the unitary group $\{\e^{ix\cdot\Phi}\}_{x\in\R^d}$.
As a consequence, we obtain that each self-adjoint operator
\begin{equation}\label{H(x)}
H_0(x):=\e^{-ix\cdot\Phi}H_0\e^{ix\cdot\Phi}
\end{equation}
has domain $\dom[H_0(x)]=\dom(H_0)$.
Similarly, the domains $\dom(\partial_jH_0)$ and $\dom(\partial_{jk}H_0)$ are left
invariant by the action of the unitary group $\{\e^{ix\cdot \Phi}\}_{x\in\R^d}$, and
the operators $(\partial_jH_0)(x):=\e^{-ix\cdot\Phi}(\partial_jH_0)\e^{ix\cdot\Phi}$
and $(\partial_{jk}H_0)(x):=\e^{-ix\cdot\Phi}(\partial_{jk}H_0)\e^{ix\cdot\Phi}$ are
self-adjoint operators with domains $\dom(\partial_jH_0)$ and $\dom(\partial_{jk}H_0)$
respectively.

Our second main assumption concerns the family of operators $H_0(x)$.

\begin{Assumption}\label{commute}
The operators $H_0(x)$, $x\in\R^d$, mutually commute.
\end{Assumption}

This assumption is equivalent to the commutativity of each $H_0(x)$ with $H_0$. As
shown in \cite[Lemma~2.4]{RT10}, Assumptions \ref{chirimoya} and \ref{commute} imply
that the operators $H_0(x),(\partial_jH_0)(y)$ and $(\partial_{k\ell}H_0)(z)$
mutually commute for each $j,k,\ell\in\{1,\ldots,d\}$ and each $x,y,z\in\R^d$. For
simplicity, we write $H_0'$ for the $d$-tuple $(\partial_1H_0,\ldots,\partial_dH_0)$,
and define for each measurable function $g:\R^d\to\C$ the operator $g(H_0')$ by using
the $d$-variables functional calculus. Similarly, we consider the family of operators $\{\partial_{jk}H_0\}$ as the components of a $d$-dimensional matrix which we denote
by $H_0''$. The symbol $E^{H_0}(\;\!\cdot\;\!)$ denotes the spectral measure of $H_0$,
and we use the notation $E^{H_0}(\lambda;\delta)$ for
$E^{H_0}\big((\lambda-\delta,\lambda+\delta)\big)$.

We now recall the definition of the critical values of $H_0$ and state some basic
properties which have been established in \cite[Lemma~2.6]{RT10}.

\begin{Definition}\label{surkappa}
A number $\lambda\in\R$ is called a critical value of $H_0$ if
\begin{equation}\label{condition}
\lim_{\varepsilon\searrow0}\big\|\big(H_0'^2+\varepsilon\big)^{-1}
E^{H_0}(\lambda;\delta)\big\|=+\infty
\end{equation}
for each $\delta>0$. We denote by $\kappa(H_0)$ the set of critical values of $H_0$.
\end{Definition}

\begin{Lemma}\label{Heigen}
Let $H_0$ satisfy Assumptions \ref{chirimoya} and \ref{commute}. Then the set
$\kappa(H_0)$ possesses the following properties:
\begin{enumerate}
\item[(a)] $\kappa(H_0)$ is closed.
\item[(b)] $\kappa(H_0)$ contains the set of eigenvalues of $H_0$.
\item[(c)] The limit\,
$
\lim_{\varepsilon\searrow0}\big\|\big(H_0'^2+\varepsilon\big)^{-1}E^{H_0}(I)\big\|
$
is finite for each compact set $I \subset \R\setminus \kappa(H_0)$.
\item[(d)] For each compact set $I\subset\R\setminus\kappa(H_0)$, there exists a
compact set $U\subset(0,\infty)$ such that $E^{H_0}(I)=E^{|H_0'|}(U)E^{H_0}(I)$.
\end{enumerate}
\end{Lemma}

In \cite[Sec.~3]{RT10} a Mourre estimate is also obtained under Assumptions
\ref{chirimoya} and \ref{commute}. It implies spectral results for $H_0$ and the
existence of locally $H_0$-smooth operators. We use the notation
$\langle x\rangle:=(1+x^2)^{1/2}$ for any $x\in\R^d$.

\begin{Theorem}\label{not_bad}
Let $H_0$ satisfy Assumptions \ref{chirimoya} and \ref{commute}. Then,
\begin{enumerate}
\item[(a)] the spectrum of $H_0$ in $\sigma(H_0)\setminus\kappa(H_0)$ is purely
absolutely continuous,
\item[(b)] each operator $B\in\B\big(\dom(\langle\Phi\rangle^{-s}),\H_0\big)$, with
$s>1/2$, is locally $H_0$-smooth on $\R\setminus\kappa(H_0)$.
\end{enumerate}
\end{Theorem}

\section{Integral formula for $\boldsymbol{H_0}$}\label{SecIntegral}
\setcounter{equation}{0}

We recall in this section the main result of \cite{RT10}, which is expressed in terms
of a function $R_f$ appearing naturally when dealing with quantum scattering theory.
The function $R_f$ is a renormalised average of a function $f$ of
localisation around the origin $0\in\R^d$. These functions were already used, in one
form or another, in \cite{GT07,RT10,Tie08,Tie09}. In these references, part of the
results were obtained under the assumption that $f$ belongs to the Schwartz space
$\S(\R^d)$. So, for simplicity, we shall assume from the very beginning that
$f\in\S(\R^d)$ and also that $f$ is even, \ie $f(x)=f(-x)$ for all $x\in\R^d$. Let us
however mention that some of the following results easily extend to the larger class
of functions introduced in \cite[Sec.~4]{RT10}.

\begin{Assumption}\label{assumption_f}
The function $f\in\S(\R^d)$ is non-negative, even and equal to $1$ on a neighbourhood
of $~0\in\R^d$.
\end{Assumption}

It is clear that $\slim_{r\to\infty}f(\Phi/r)=1$ if $f$ satisfies Assumption
\ref{assumption_f}. Furthermore, it also follows from this assumption that the
function $R_f:\R^d\setminus\{0\}\to\R$ given by
$$
R_f(x):=\int_0^\infty\frac{\d\mu}\mu\big(f(\mu x)-\chi_{[0,1]}(\mu)\big)
$$
is well-defined. The following properties of $R_f$ are proved in
\cite[Sec. 2]{Tie09}: The function $R_f$ belongs to $C^\infty(\R^d\setminus\{0\})$
and satisfies
\begin{equation*}
R_f'(x)=\int_0^\infty\d\mu\;\!f'(\mu x)
\end{equation*}
as well as the homogeneity properties $x\cdot R_f'(x)=-1$ and
$t^{|\alpha|}(\partial^\alpha R_f)(tx)=(\partial^\alpha R_f)(x)$, where
$\alpha\in\N^d$ is a multi-index and $t>0$. Furthermore, if $f$ is radial, then
$R_f'(x)=-x^{-2}x$. We shall also need the function $F_f:\R^d\setminus\{0\}\to\R$
defined by
\begin{equation}\label{def_F}
F_f(x):=\int_\R\d\mu\;\!f(\mu x).
\end{equation}
The function $F_f$ satisfies several properties as $R_f$ such as $F_f(x)=tF_f(tx)$
for each $t>0$ and each $x\in\R^d\setminus\{0\}$.

Now, we know from Lemma \ref{Heigen}.(a) that the set $\kappa(H_0)$ is closed. So we
can define for each $t\ge0$ the set
$$
\D_t:=\big\{\varphi\in\dom(\langle\Phi\rangle^t)\mid\varphi=\eta(H_0)\varphi
\textrm{ for some }\eta\in C^\infty_{\rm c}\big(\R\setminus\kappa(H_0)\big)\big\}.
$$
The set $\D_t$ is included in the subspace $\H_{\rm ac}(H_0)$ of absolute
continuity of $H_0$, due to Theorem \ref{not_bad}.(a), and
$\D_{t_1}\subset\D_{t_2}$ if $t_1\ge t_2$. We refer the reader to
\cite[Sec.~6]{RT10} for an account on density properties of the sets $\D_t$.

In the sequel, we sometimes write $C^{-1}$ for an operator $C$ a priori not
invertible. In such a case, the operator $C^{-1}$ will always be acting on a set
where it is well-defined. Next statement follows from \cite[Prop.~5.2]{RT10} and
\cite[Rem.~5.4]{RT10}.

\begin{Proposition}\label{lemma_T_f}
Let $H_0$ satisfy Assumptions \ref{chirimoya} and \ref{commute}, and let $f$ satisfy
Assumption \ref{assumption_f}. Then the map
$$
t_f:\D_1\to\C,\quad\v\mapsto
t_f(\v):=-\12\sum_j\big\{\big\langle\Phi_j\v,(\partial_jR_f)(H_0')\v\big\rangle
+\big\langle\big(\partial_jR_f \big)(H_0')\v,\Phi_j\v\big\rangle\big\},
$$
is well-defined. Moreover, the linear operator $T_f:\D_1\to\H_0$ defined by
\begin{equation}\label{nemenveutpas}
\textstyle
T_f\v:=-\12\big(\Phi\cdot R_f'(H_0')+R_f'\big(\frac{H_0'}{|H_0'|})\cdot\Phi\;\!|H_0'|^{-1}
+iR_f'\big(\frac{H_0'}{|H_0'|}\big)\cdot\big(H_0''^{\sf T}H_0'\big)|H_0'|^{-3}\big)\v
\end{equation}
satisfies $t_f(\v)=\langle\v,T_f\v\rangle$ for each $\v\in\D_1$. In particular,
$T_f$ is a symmetric operator if $\D_1$ is dense in $\H_0$.
\end{Proposition}

\begin{Remark}\label{RemarkOnTf}
Formula \eqref{nemenveutpas} is a priori rather complicated and one could be tempted
to replace it by the simpler formula
$-\12\big(\Phi\cdot R_f'(H_0')+R_f'(H_0')\cdot\Phi\big)$. Unfortunately, a precise
meaning of this expression is not available in general, and its full derivation can
only be justified in concrete examples. However, when $f$ is radial, then
$(\partial_jR_f)(x)=-x^{-2}x_j$, and $T_f$ is equal on $\D_1$ to
\begin{equation}
\textstyle
T:=\12 \big(\Phi\cdot\frac{H_0'}{(H_0')^2}+\frac{H_0'}{|H_0'|}\cdot\Phi\;\!|H_0'|^{-1}
+\frac{iH_0'}{(H_0')^4}\cdot\big(H_0''^{\sf T}H_0'\big)\big).\label{T}
\end{equation}
\end{Remark}

Next theorem is the main result of \cite{RT10}; it relates the evolution of the
localisation operators $f(\Phi/r)$ to the operator $T_f$.

\begin{Theorem}[Theorem 5.5 of \cite{RT10}]\label{for_Schwartz}
Let $H_0$ satisfy Assumptions \ref{chirimoya} and \ref{commute}, and let $f$ satisfy
Assumption \ref{assumption_f}. Then we have for each $\v\in\D_2$
\begin{equation}\label{T_f}
\lim_{r\to\infty}\12\int_0^\infty\d t\,\big\langle\v,
\big(\e^{-itH_0}f(\Phi/r)\e^{itH_0}-\e^{itH_0}f(\Phi/r)\e^{-itH_0}\big)\v\big\rangle\\
=\langle\v,T_f\v\rangle.
\end{equation}
\end{Theorem}

In particular, when the localisation function $f$ is radial, the  operator $T_f$ in
the r.h.s. of \eqref{T_f} is equal to the operator $T$, which is independent of $f$.

\section{Symmetrized time delay}\label{SecSymDelay}
\setcounter{equation}{0}

In this section we prove the existence of symmetrized time delay for a scattering
system $(H_0,H,J)$ with free operator $H_0$, full operator $H$, and identification
operator $J$. The operator $H_0$ acts in the Hilbert space $\H_0$ and satisfies the
assumptions \ref{chirimoya} and \ref{commute} with respect to the family $\Phi$. The
operator $H$ is a self-adjoint operator in a Hilbert space $\H$ satisfying the
assumption \ref{wave} below. The operator $J:\H_0\to\H$ is a bounded operator used to
``identify" the Hilbert space $\H_0$ with a subset of $\H$.

The assumption on $H$ concerns the existence, the isometry and the completeness of
the generalised wave operators:

\begin{Assumption}\label{wave}
The generalised wave operators
$$
\displaystyle
W_\pm:=\slim_{t\to\pm\infty}\e^{itH}J\e^{-itH_0}P_{\rm ac}(H_0)
$$
exist, are partial isometries with initial subspaces $\H_0^\pm$ and final subspaces
$\H_{\rm ac}(H)$.
\end{Assumption}

Sufficient conditions on $JH_0-HJ$ ensuring the existence and the completeness of
$W_\pm$ are given in \cite[Chap.~5]{Yaf92}. The main consequence of Assumption
\ref{wave} is that the scattering operator
\begin{equation*}
S:= W_+^*W_-:\H_0^-\to\H_0^+
\end{equation*}
is a well-defined unitary operator commuting with $H_0$.

We now define the sojourn times for the quantum scattering system $(H_0,H,J)$,
starting with the sojourn time for the free evolution $\e^{-itH_0}$. So, let $r>0$
and let $f$ be a non-negative element of $\S(\R^d)$ equal to $1$ on a neighbourhood
$\Sigma$ of the origin $0\in\R^d$. For $\varphi\in\D_0$, we set
$$
T_r^0(\varphi)
:=\int_\R\d t\,\big\langle\e^{-itH_0}\varphi,f(\Phi/r)\e^{-itH_0}\varphi\big\rangle,
$$
where the integral has to be understood as an improper Riemann integral. The operator
$f(\Phi/r)$ is approximately the projection onto the subspace $E^\Phi(r\Sigma)\H_0$
of $\H_0$, with $r\Sigma:=\{x\in\R^d\mid x/r\in\Sigma\}$. Therefore, if
$\|\varphi\|=1$, then $T_r^0(\varphi)$ can be approximately interpreted as the time
spent by the evolving state $\e^{-itH_0}\varphi$ inside $E^\Phi(r\Sigma)\H_0$.
Furthermore, the expression $T_r^0(\varphi)$ is finite for each $\varphi\in\D_0$,
since we know from Lemma \ref{not_bad}.(b) that each operator
$B\in\B\big(\dom(\langle\Phi\rangle^{-s}),\H_0\big)$, with $s>\12$, is locally
$H_0$-smooth on $\R\setminus\kappa(H_0)$.

When defining the sojourn time for the full evolution $\e^{-itH}$, one faces the
problem that the localisation operator $f(\Phi/r)$ acts in $\H_0$ while the operator
$\e^{-itH}$ acts in $\H$. The obvious modification would be to consider the operator
$Jf(\Phi/r)J^*\in\B(\H)$, but the resulting framework could be not general enough
(see Remark \ref{RemJ} below). Sticking to the basic idea that the freely evolving
state $\e^{-itH_0}\varphi$ should approximate, as $t\to\pm\infty$, the corresponding
evolving state $\e^{-itH}W_\pm\varphi$, one should look for operators
$L(t):\H\to\H_0$, $t\in\R$, such that
\begin{equation}\label{defconv}
\lim_{t\to\pm\infty}\left\|L(t)\e^{-itH}W_\pm\varphi-\e^{-itH_0}\varphi\right\|=0.
\end{equation}
Since we consider vectors $\varphi\in\D_0$, the operators $L(t)$ can be unbounded as
long as $L(t)E^{H}(I)$ are bounded for any bounded subset $I\subset\R$. With such a
family of operators $L(t)$, it is natural to define the sojourn time for the full
evolution $\e^{-itH}$ by the expression
\begin{equation}\label{term1}
T_{r,1}(\varphi):=\int_\R\d t\,
\big\langle L(t)\e^{-itH}W_-\varphi,f(\Phi/r)L(t)\e^{-itH}W_-\varphi\big\rangle.
\end{equation}
Another sojourn time appearing naturally in this context is
\begin{equation}\label{term2}
T_2(\varphi):=\int_\R\d t\,\big\langle\e^{-itH}W_-\varphi,\big(1-L(t)^*L(t)\big)
\e^{-itH}W_-\varphi\big\rangle_\H.
\end{equation}
The finiteness of $T_{r,1}(\varphi)$ and $T_2(\varphi)$ is proved under an additional
assumption in Lemma \ref{lemma_free} below. The term $T_{r,1}(\varphi)$ can be
approximatively interpreted as the time spent by the scattering state
$\e^{-itH}W_-\varphi$, injected in $\H_0$ via $L(t)$, inside $E^\Phi(r\Sigma)\H_0$.
The term $T_2(\varphi)$ can be seen as the time spent by the scattering state
$\e^{-itH}W_-\varphi$ inside the time-dependent subset $\big(1-L(t)^*L(t)\big)\H$ of
$\H$. If $L(t)$ is considered as a time-dependent quasi-inverse for the
identification operator $J$ (see \cite[Sec.~2.3.2]{Yaf92} for the related
time-independent notion of quasi-inverse), then the subset $\big(1-L(t)^*L(t)\big)\H$
can be seen as an approximate complement of $J\H_0$ in $\H$ at time $t$. When
$\H_0=\H$, one usually sets $L(t)=J^*=1$, and the term $T_2(\varphi)$ vanishes. Within
this general framework, we say that
\begin{equation}\label{symsym}
\tau_r(\varphi):=T_r(\varphi)-\12\big\{T_r^0(\varphi)+T_r^0(S\varphi)\big\},
\end{equation}
with $T_r(\varphi):=T_{r,1}(\varphi)+T_2(\varphi)$, is the symmetrized time delay of
the scattering system $(H_0,H,J)$ with incoming state $\varphi$. This symmetrized
version of the usual time delay
\begin{equation*}
\tau_r^{\rm in}(\varphi):=T_r(\varphi)-T_r^0(\varphi)
\end{equation*}
is known to be the only time delay having a well-defined limit as $r\to\infty$ for
complicated scattering systems (see for example
\cite{AJ07,BO79,GT07,Mar75,Mar81,SM92,Smi60,Tie06}).

For the next lemma, we need the auxiliary quantity
\begin{equation}\label{tau_free}
\tau^{\rm free}_r(\varphi):=\12\int_0^\infty\d t\,
\big\langle\varphi,S^*\big[\e^{itH_0}f(\Phi/r)\e^{-itH_0}
-\e^{-itH_0}f(\Phi/r)\e^{itH_0},S\big]\varphi\big\rangle,
\end{equation}
which is finite for all $\varphi\in\H_0^-\cap\D_0$. We refer the reader to
\cite[Eq.~(4.1)]{Tie09} for a similar definition in the case of dispersive systems,
and to \cite[Eq.~(3)]{AC87}, \cite[Eq.~(6.2)]{Jen81} and \cite[Eq.~(5)]{Mar76} for
the original definition.

\begin{Lemma}\label{lemma_free}
Let $H_0$, $f$ and $H$ satisfy Assumptions \ref{chirimoya}, \ref{commute},
\ref{assumption_f} and \ref{wave}, and let $\varphi\in\H_0^-\cap\D_0$ be such that
\begin{equation}\label{H-+}
\big\|\big(L(t)W_--1\big)\e^{-itH_0}\varphi\big\|\in\lone(\R_-,\d t)
\qquad\hbox{and}\qquad\big\|(L(t)W_+-1)\e^{-itH_0}S\varphi\big\|\in\lone(\R_+,\d t).
\end{equation}
Then $T_r(\varphi)$ is finite for each $r>0$, and
\begin{equation}\label{equality}
\lim_{r\to\infty}\big\{\tau_r(\varphi)-\tau^{\rm free}_r(\varphi)\big\}=0.
\end{equation}
\end{Lemma}

\begin{proof}
Direct computations with $\varphi\in\H_0^-\cap\D_0$ imply that
\begin{align*}
I_r(\varphi)&:=T_{r,1}(\varphi)-\12\big\{T_r^0(\varphi)+T_r^0(S\varphi)\big\}
-\tau^{\rm free}_r(\varphi)\\
&=\int_{-\infty}^0\d t\,\big\{\big\langle L(t)\e^{-itH}W_-\varphi,
f(\Phi/r)L(t)\e^{-itH}W_-\varphi\big\rangle
-\big\langle\e^{-itH_0}\varphi,f(\Phi/r)\e^{-itH_0}\varphi\big\rangle\big\}\\
&\quad+\int_0^\infty\d t\,\big\{\big\langle L(t)\e^{-itH}W_-\varphi,
f(\Phi/r)L(t)\e^{-itH}W_-\varphi\big\rangle
-\big\langle\e^{-itH_0}S\varphi,f(\Phi/r)\e^{-itH_0}S\varphi\big\rangle\big\}.
\end{align*}
Using the inequality
$$
\big|\|\varphi\|^2-\|\psi\|^2\big|
\le\|\varphi-\psi\|\cdot\big(\|\varphi\|+\|\psi\|\big),
\quad\varphi,\psi\in\H_0,
$$
the intertwining property of the wave operators and the identity $W_-=W_+ S$, one
gets the estimates
\begin{align*}
\big|\big\langle L(t)\e^{-itH}W_-\varphi,f(\Phi/r)L(t)\e^{-itH}W_-\varphi\big\rangle
-\big\langle\e^{-itH_0}\varphi,f(\Phi/r)\e^{-itH_0}\varphi\big\rangle\big|
&\le{\rm Const.}\,g_-(t),\\
\big|\big\langle L(t)\e^{-itH}W_-\varphi,f(\Phi/r)L(t)\e^{-itH}W_-\varphi\big\rangle
-\big\langle\e^{-itH_0}S\varphi,f(\Phi/r)\e^{-itH_0}S\varphi\big\rangle\big|
&\le{\rm Const.}\,g_+(t),
\end{align*}
where
$$
g_-(t):=\big\|\big(L(t)W_--1\big)\e^{-itH_0}\varphi\big\|\qquad\hbox{and}\qquad
g_+(t):=\big\|\big(L(t)W_+-1\big)\e^{-itH_0}S\varphi\big\|.
$$
It follows by \eqref{H-+} that $|I_r(\varphi)|$ is bounded by a constant independent
of $r$, and thus $T_{r,1}(\varphi)$ is finite for each $r>0$. Then, using Lebesgue's
dominated convergence theorem, the fact that $\slim_{r\to\infty}f(\Phi/r)=1$ and the
isometry of $W_-$ on $\H_0^-$, one obtains that
\begin{align*}
\lim_{r\to\infty}I_r(\varphi)
&=
\int_{-\infty}^0\d t\,\big\{\big\langle L(t)\e^{-itH}W_-\varphi,
L(t)\e^{-itH}W_-\varphi\big\rangle
-\big\langle\e^{-itH_0}\varphi,\e^{-itH_0}\varphi\big\rangle\big\}\\
&\quad+\int_0^\infty\d t\,\big\{\big\langle L(t )\e^{-itH}W_-\varphi,
L(t)\e^{-itH}W_-\varphi\big\rangle
-\big\langle\e^{-itH_0}S\varphi,\e^{-itH_0}S\varphi\big\rangle\big\}\\
&=\int_\R\d t\,\big\langle\e^{-itH}W_-\varphi,
\big(L(t)^*L(t)-1\big)\e^{-itH}W_-\varphi\big\rangle_\H\\
&\equiv-T_2(\varphi).
\end{align*}
Thus, $T_2(\varphi)$ is finite, and the equality \eqref{equality} is verified. Since
$T_r(\varphi)=T_{r,1}(\varphi)+T_2(\varphi)$, one also infers that $T_r(\varphi)$ is
finite for each $r>0$.
\end{proof}

Next Theorem shows the existence of the symmetrized time delay. It is a direct
consequence of Lemma \ref{lemma_free}, Definition \eqref{tau_free} and Theorem
\ref{for_Schwartz}. The apparently large number of assumptions reflects nothing more
but the need of describing the very general scattering system $(H_0,H,J)$; one needs
hypotheses on the relation between $H_0$ and $\Phi$, conditions on the localisation
function $f$, a compatibility assumption between $H_0$ and $H$, and conditions on the
state $\varphi$ on which the calculation are performed.

\begin{Theorem}\label{sym_case}
Let $H_0$, $f$ and $H$ satisfy Assumptions \ref{chirimoya}, \ref{commute},
\ref{assumption_f} and \ref{wave}, and let $\varphi\in\H_0^-\cap\D_2$ satisfy
$S\varphi\in\D_2$ and \eqref{H-+}. Then one has
\begin{equation}\label{Eisenbud_sym}
\lim_{r\to\infty}\tau_r(\varphi)
=-\big\langle\varphi,S^*\big[T_f,S\big]\varphi\big\rangle,
\end{equation}
with $T_f$ defined by \eqref{nemenveutpas}.
\end{Theorem}

\begin{Remark}
Theorem \ref{sym_case} is the main result of the paper. It expresses the identity of
the symmetrized time delay (defined in terms of sojourn times) and the Eisenbud-Wigner
time delay for general scattering systems $(H_0,H,J)$. The l.h.s. of
\eqref{Eisenbud_sym} is equal to the global symmetrized time delay of the scattering
system $(H_0,H,J)$, with incoming state $\varphi$, in the dilated regions associated
to the localisation operators $f(\Phi/r)$. The r.h.s. of \eqref{Eisenbud_sym} is the
expectation value in $\varphi$ of the generalised Eisenbud-Wigner time delay operator
$-S^*[T_f,S]$. When $T_f$ acts in the spectral representation of $H_0$ as the
differential operator $i\frac\d{\d H_0}$, which occurs in most of the situations of
interest (see for example \cite[Sec.~7]{RT10}), one recovers the usual Eisenbud-Wigner
Formula:
$$
\lim_{r\to\infty}\tau_r(\varphi)
\textstyle=-\big\langle\varphi,iS^*\frac{\d S}{\d H_0}\;\!\varphi\big\rangle.
$$
\end{Remark}

\begin{Remark}\label{RemJ}
Equation \eqref{defconv} is equivalent to the existence of the limits
\begin{equation*}
\widetilde W_\pm:=\slim_{t\to\pm\infty}\e^{itH_0}L(t)\e^{-itH}P_{\rm ac}(H),
\end{equation*}
together with the equalities $\widetilde W_\pm W_\pm =P_0^\pm$, where $P_0^\pm$ are
the orthogonal projections on the subspaces $\H^\pm_0$ of $\H_0$. In simple
situations, namely, when $\H^\pm_0=\H_{\rm ac}(H_0)$ and $L(t)\equiv L$ is
independent of $t$ and bounded, sufficient conditions implying \eqref{defconv} are
given in \cite[Thm.~ 2.3.6]{Yaf92}. In more complicated situations, namely, when
$\H^\pm_0\neq\H_{\rm ac}(H_0)$ or $L(t)$ depends on $t$ and is unbounded, the proof
of \eqref{defconv} could be highly non-trivial. This occurs for instance in the case
of the $N$-body systems. In such a situation, the operators $L(t)$ really depend on
$t$ and are unbounded (see for instance \cite[Sec.~6.7]{DG97}), and the proof of
\eqref{defconv} is related to the problem of the asymptotic completeness of the
$N$-body systems.
\end{Remark}

\section{Usual time delay}\label{Sec_Usual}
\setcounter{equation}{0}

We give in this section conditions under which the symmetrized time delay
$\tau_r(\varphi)$ and the usual time delay $\tau_r^{\rm in}(\varphi)$ are equal in the
limit $r\to \infty$. Heuristically, one cannot expect that this equality holds if the
scattering is not elastic or is of multichannel type. However, for simple scattering
systems, the equality of both time delays presents an interest. At the mathematical
level, this equality reduces to giving conditions under which
\begin{equation}\label{clementine}
\lim_{r\to\infty}\big\{T_r^0(S\varphi)-T_r^0(\varphi)\big\}=0.
\end{equation}
Equation \eqref{clementine} means that the freely evolving states
$\e^{-itH_0}\varphi$ and $\e^{-itH_0}S\varphi$ tend to spend the same time within the
region defined by the localisation function $f(\Phi/r)$ as $r\to\infty$. Formally,
the argument goes as follows. Suppose that $F_f(H_0')$, with $F_f$ defined in
\eqref{def_F}, commutes with the scattering operator $S$. Then, using the change of
variables $\mu:=t/r$, $\nu:=1/r$, and the symmetry of $f$, one gets
\begin{align*}
\lim_{r\to\infty}\big\{T_r^0(S\varphi)-T_r^0(\varphi)\big\}
&=\lim_{r\to\infty}\int_\R\d t\,
\big\langle\varphi,S^*[\e^{itH_0}f(\Phi/r)\e^{-itH_0},S]\varphi\big\rangle
-\<\varphi,S^*[F_f(H_0'),S]\varphi\>\\
&=\lim_{\nu\searrow0}\int_\R\d\mu\,
\big\langle\varphi,S^*\big[\textstyle\frac1\nu
\big\{f(\mu H_0'+\nu\Phi)-f(\mu H_0')\big\},S\big]\varphi\big\rangle\\
&=\int_\R\d\mu\,
\big\langle\varphi,S^*[\Phi\cdot f'(\mu H_0'),S]\varphi\big\rangle\\
&=0.
\end{align*}
A rigorous proof of this argument is given in Theorem \ref{equal_sojourn} below.
Before this we introduce an assumption on the behavior of the $C_0$-group
$\{\e^{ix\cdot\Phi}\}_{x\in\R^d}$ in $\dom(H_0)$, and then prove a technical lemma.
We use the notation $\G$ for $\dom(H_0)$ endowed with the graph topology, and $\G^*$
for its dual space. In the following proofs, we also freely use the notations of
\cite{ABG} for some regularity classes with respect to the group generated by $\Phi$.

\begin{Assumption}\label{chirimoya2}
The $C_0$-group $\{\e^{ix\cdot\Phi}\}_{x\in\R^d}$ is of polynomial growth in $\G$,
namely there exists $r>0$ such that for all $x\in\R^d$
\begin{equation*}
\left\|\e^{ix\cdot \Phi}\right\|_{\B(\G,\G)}\leq{\rm Const.}\;\!\langle x\rangle^r.
\end{equation*}
\end{Assumption}

\begin{Lemma}\label{batman}
Let $H_0$ and $\Phi$ satisfy Assumptions \ref{chirimoya}, \ref{commute} and
\ref{chirimoya2}, and let $\eta\in C^\infty_{\rm c}(\R)$. Then there exists
$\textsc c,s>0$ such that for all $\mu\in\R$, $x\in\R^d$ and
$\nu\in(-1,1)\setminus\{0\}$
$$
\big\|\textstyle\frac1\nu\big\{\eta\big(H_0(\nu x)\big)
\e^{i\frac\mu\nu[H_0(\nu x)-H_0]}-\eta(H_0)\e^{i\mu x\cdot H_0'}\big\}\big\|
\le\textsc c\;\!(1+|\mu|)\langle x\rangle^s.
$$
\end{Lemma}

\begin{proof}
For $x\in\R^d$ and $\mu\in\R$, we define the function
$$
g_{x,\mu}:(-1,1)\setminus\{0\}\to\B(\H_0),\quad
\nu\mapsto\e^{i\frac\mu\nu[H_0(\nu x)-H_0]}\eta(H_0).
$$
Reproducing the argument of point (ii) of the proof of \cite[Thm.~5.5]{RT10}, one
readily shows that $H_0\in C_{\rm u}^1(\Phi;\G,\H_0)$, and then that $g_{x,\mu}$ is
continuous with
$$
g_{x,\mu}(0):=\lim_{\nu\to 0} g_{x,\mu}(\nu)=\e^{i\mu x\cdot H_0'}\eta(H_0).
$$
On another hand, since $\eta(H_0)$ belongs to $C^1_{\rm u}(\Phi)$, one has in
$\B(\H_0)$ the equalities
$$
\textstyle\frac1\nu\big\{\eta\big(H_0(\nu x)\big)-\eta(H_0)\big\}
=\frac1\nu\displaystyle\int_0^1\d t\,\frac\d{\d t}\;\!\eta\big(H_0(t\nu x)\big)
=i\sum_jx_j\int_0^1\d t\,\e^{-it\nu x\cdot\Phi}\big[\eta(H_0),\Phi_j\big]
\e^{it\nu x\cdot\Phi}.
$$
So, combining the two equations, one obtains that
\begin{align}
&\textstyle\frac1\nu\big\{\eta\big(H_0(\nu x)\big)\e^{i\frac\mu\nu[H_0(\nu x)-H_0]}
-\eta(H_0)\e^{i\mu x\cdot H_0'}\big\}\nonumber\\
&=\textstyle\frac1\nu\big\{\eta\big(H_0(\nu x)\big)-\eta(H_0)\big\}
\e^{i\frac\mu\nu[H_0(\nu x)-H_0]}+\frac1\nu\big\{g_{x,\mu}(\nu)-g_{x,\mu}(0)\big\}
\nonumber\\
&=i\sum_jx_j\int_0^1\d t\,\e^{-it\nu x\cdot\Phi}\big[\eta(H_0),\Phi_j\big]
\e^{it\nu x\cdot\Phi}\e^{i\frac\mu\nu[H_0(\nu x)-H_0]}
+\textstyle\frac1\nu\big\{g_{x,\mu}(\nu)-g_{x,\mu}(0)\big\}.\label{crocro}
\end{align}
In order to estimate the difference $g_{x,\mu}(\nu)-g_{x,\mu}(0)$, observe first that
one has in $\B(\H_0)$ for any bounded set $I\subset\R$
$$
\textstyle\frac1\nu\big[H_0(\nu x)-H_0\big]E^{H_0}(I)
=\frac1\nu\displaystyle\int_0^1\d t\,\frac\d{\d t}\;\!H_0(t\nu x)E^{H_0}(I)
=\int_0^1\d t\,x\cdot H'_0(t\nu x)E^{H_0}(I).
$$
So, if $\varepsilon\in\R$ is small enough and if the bounded set $I\subset\R$ is
chosen such that $\eta(H_0)=E^{H_0}(I)\eta(H_0)$, one obtains in $\B(\H_0)$
\begin{align*}
&g_{x,\mu}(\nu+\varepsilon)-g_{x,\mu}(\nu)\\
&=\big\{\e^{i\mu\int_0^1\d t\,x\cdot H'_0(t(\nu+\varepsilon)x)E^{H_0}(I)}
-\e^{i\mu\int_0^1\d t\,x\cdot H'_0(t\nu x)E^{H_0}(I)}\big\}\eta(H_0)\\
&=\e^{i\mu\int_0^1\d u\,x\cdot H'_0(u\nu x)E^{H_0}(I)}
\big\{\e^{i\mu\int_0^1\d t\,x\cdot[H'_0(t(\nu+\varepsilon)x)
-H'_0(t\nu x)]E^{H_0}(I)}-1\big\}\eta(H_0)\\
&=\e^{i\mu\int_0^1\d u\,x\cdot H'_0(u\nu x)E^{H_0}(I)}
\big\{\e^{i\mu\int_0^1\d t\int_0^1\d s\,t\varepsilon\sum_{j,k}x_jx_k
(\partial_{jk}H_0)(t(\nu+s\varepsilon)x)E^{H_0}(I)}-1\big\}\eta(H_0).
\end{align*}
Note that the property $\partial_jH_0\in C^1_{\rm u}(\Phi;\G,\H_0)$ (which follows
from Assumption \ref{chirimoya} and \cite[Lemma~5.1.2.(b)]{ABG}) has been taken into
account for the last equality. Then, multiplying the above expression by
$\varepsilon^{-1}$ and taking the limit $\varepsilon\to0$ in $\B(\H_0)$ leads to
\begin{equation}
g'_{x,\mu}(\nu)
=i\mu\e^{i\mu\int_0^1\d u\,x\cdot H'_0(u\nu x)}
\int_0^1\d t\,t\sum_{j,k}x_j x_k(\partial_{jk}H_0)(t\nu x)\;\!\eta(H_0).
\label{etencore}
\end{equation}
This formula, together with Equation \eqref{crocro} and the mean value theorem,
implies that
\begin{align}
&\textstyle\big\|\frac1\nu\big\{\eta\big(H_0(\nu x)\big)
\e^{i\frac\mu\nu[H_0(\nu x)-H_0]}-\eta(H_0)\e^{i\mu x\cdot H_0'}\big\}\big\|
\nonumber\\
&\le{\rm Const.}\;\!|x|+\sup_{\xi\in[0,1]}\big\|g_{x,\mu}'(\xi\nu)\big\|
\nonumber\\
&\le{\rm Const.}\;\!|x|+{\rm Const.}\,x^2|\mu|\sup_{\xi\in[0,1]}\sum_{j,k}
\big\|(\partial_{jk}H_0)(\xi\nu x)\;\!\eta(H_0)\big\|.\label{bound_D}
\end{align}
But one has
\begin{equation*}
(\partial_{jk} H_0)(\xi\nu x)\;\!\eta(H_0)
=\e^{-i\xi\nu x\cdot \Phi}(\partial_{jk} H_0)\e^{i\xi\nu x\cdot \Phi}\eta(H_0)
\end{equation*}
with $\eta(H_0)\in\B(\H_0,\G)$ and $(\partial_{jk}H_0)\in \B(\G,\H_0)$. So, it
follows from Assumption \ref{chirimoya2} that there exists $r>0$ such that
$$
\big\|(\partial_{jk}H_0)(\xi\nu x)\;\!\eta(H_0)\big\|
\leq{\rm Const.}\;\!\langle\xi\nu x\rangle^r.
$$
Hence, one finally gets from \eqref{bound_D} that for each
$\nu\in(-1,1)\setminus\{0\}$
$$
\textstyle\big\|\frac1\nu\big\{\eta\big(H_0(\nu x)\big)
\e^{i\frac\mu\nu[H_0(\nu x)-H_0]}-\eta(H_0)\e^{i\mu x\cdot H_0'}\big\}\big\|
\le{\rm Const.}(1+|\mu|)\langle x\rangle^{r+2},
$$
which proves the claim with $s:=r+2$.
\end{proof}

In the sequel, the symbol $\F$ stands for the Fourier transformation, and the measure
$\underline\d x$ on $\R^d$ is chosen so that $\F$ extends to a unitary operator in
$\ltwo(\R^d)$.

\begin{Theorem}\label{equal_sojourn}
Let $H_0,f,H$ and $\Phi$ satisfy Assumptions \ref{chirimoya}, \ref{commute},
\ref{assumption_f}, \ref{wave} and \ref{chirimoya2}, and let $\varphi\in\H_0^- \cap \D_2$ satisfy
$S\varphi\in\D_2$ and
\begin{equation}\label{grobobo}
\big[F_f(H'_0),S\big]\varphi=0.
\end{equation}
Then the following equality holds:
$$
\lim_{r\to\infty}\big\{T_r^0(S\varphi)-T_r^0(\varphi)\big\}=0.
$$
\end{Theorem}

Note that the l.h.s. of \eqref{grobobo} is well-defined due to the homogeneity
property of $F_f$. Indeed, one has
$$
\big[F_f(H'_0),S\big]\varphi
=\big[|H'_0|^{-1}\eta(H_0)F_f\big(\textstyle\frac{H'_0}{|H'_0|}\big),S\big]\varphi
$$
for some $\eta\in C^\infty_{\rm c}\big(\R\setminus\kappa(H_0)\big)$, and thus
$\big[F_f(H'_0),S\big]\varphi\in\H$ due to Lemma \ref{Heigen}.(d) and the compacity
of $F_f(\mathbb S^{d-1})$.

\begin{proof}
Let $\varphi\in\H_0^-\cap \D_2$ satisfies $S\varphi\in\D_2$, take a real
$\eta\in C^\infty_{\rm c}\big(\R\setminus\kappa(H_0)\big)$ such that
$\varphi=\eta(H_0)\varphi$, and set $\eta_t(H_0):=\e^{itH_0}\eta(H_0)$. Using
\eqref{grobobo}, the definition of $F_f$ and the change of variables $\mu:=t/r$,
$\nu:=1/r$, one gets
\begin{align}
&T_{1/\nu}^0(S\varphi)-T_{1/\nu}^0(\varphi)\nonumber\\
&=\int_\R\d\mu\,\big\langle\varphi,S^*\big[\textstyle\frac1\nu
\big\{\eta_{\frac\mu\nu}(H_0)f(\nu \Phi)\eta_{-\frac\mu\nu}(H_0)
-f(\mu H'_0)\big\},S\big]\varphi\big\rangle\nonumber\\
&=\int_\R\d\mu\int_{\R^d}\underline\d x\,(\F f)(x)
\big\langle\varphi,S^*\big[\textstyle\frac1\nu\big\{\e^{i\nu x\cdot \Phi}
\eta_{\frac\mu\nu}\big(H_0(\nu x)\big)\eta_{-\frac\mu\nu}(H_0)
-\e^{i\mu x\cdot H_0'}\big\},S\big]\varphi\big\rangle\nonumber\\
&=\int_\R\d\mu\int_{\R^d}\underline\d x\,(\F f)(x)
\big\langle\varphi,S^*\big[\textstyle\frac1\nu(\e^{i\nu x\cdot \Phi}-1)
\eta\big(H_0(\nu x)\big)\e^{i\frac\mu\nu[H_0(\nu x)-H_0]},S\big]\varphi\big\rangle
\label{first_term}\\
&\quad+\int_\R\d\mu\int_{\R^d}\underline\d x\,(\F f)(x)
\big\langle\varphi,S^*\big[\textstyle\frac1\nu
\big\{\eta\big(H_0(\nu x)\big)\e^{i\frac\mu\nu[H_0(\nu x)-H_0]}
-\eta(H_0)\e^{i\mu x\cdot H'_0}\big\},S\big]\varphi\big\rangle.\nonumber
\end{align}
To prove the statement, it is sufficient to show that the limit as $\nu\searrow0$ of
each of these two terms is equal to zero. This is done in points (i) and (ii) below.

(i) For the first term, one can easily adapt the method \cite[Thm.~5.5]{RT10}
(points (ii) and (iii) of the proof) in order to apply Lebesgue's dominated
convergence theorem to \eqref{first_term}. So, one gets
\begin{align*}
&\lim_{\nu\searrow0}\int_\R\d\mu\int_{\R^d}\underline\d x\,(\F f)(x)
\big\langle\varphi,S^*\big[\textstyle\frac1\nu(\e^{i\nu x\cdot \Phi}-1)
\eta\big(H_0(\nu x)\big)\e^{i\frac\mu\nu[H_0(\nu x)-H_0]},S\big]\varphi\big\rangle\\
&=i\int_\R\d\mu\int_{\R^d}\underline\d x\,(\F f)(x)
\big\{\big\langle(x\cdot \Phi)S\varphi,\e^{i\mu x\cdot H'_0}S\varphi\big\rangle
-\big\langle(x\cdot \Phi)\varphi,\e^{i\mu x\cdot H'_0}\varphi\big\rangle\big\},
\end{align*}
and the change of variables $\mu':=-\mu$, $x':=-x$, together with the symmetry of
$f$, implies that this expression is equal to zero.

(ii) For the second term, it is sufficient to prove that
\begin{equation}\label{second_term}
\lim_{\nu\searrow0}\int_\R\d\mu\int_{\R^d}\underline\d x\,(\F f)(x)
\big\langle\psi,\textstyle\frac1\nu\big\{\eta\big(H_0(\nu x)\big)
\e^{i\frac\mu\nu[H_0(\nu x)-H_0]}-\eta(H_0)\e^{i\mu x\cdot H'_0}\big\}
\psi\big\rangle
\end{equation}
is equal to zero for any $\psi \in \D_2$ satisfying $\eta(H_0)\psi=\psi$. For
the moment, let us assume that we can interchange the limit and the integrals
in \eqref{second_term} by invoking Lebesgue's dominated convergence theorem. Then,
taking Equations \eqref{crocro} and \eqref{etencore} into account, one obtains
\begin{align*}
&\lim_{\nu\searrow0}\int_\R\d\mu\int_{\R^d}\underline\d x\,(\F f)(x)
\big\langle\psi,\textstyle\frac1\nu\big\{\eta\big(H_0(\nu x)\big)
\e^{i\frac\mu\nu[H_0(\nu x)-H_0]}-\eta(H_0)\e^{i\mu x\cdot H'_0}\big\}
\psi\big\rangle\\
&=\int_\R\d\mu\int_{\R^d}\underline\d x\,(\F f)(x)
\big\langle\psi,\big\{i\big[\eta(H_0),x\cdot\Phi\big]\e^{i\mu x\cdot H_0'}
+\textstyle\frac{i\mu}2\e^{i\mu x\cdot H'_0} \sum_{j,k}x_jx_k(\partial_{jk}H_0)
\;\!\eta(H_0)\big\}\psi\big\rangle,
\end{align*}
and the change of variables $\mu':=-\mu$, $x':=-x$, together with the symmetry of
$f$, implies that this expression is equal to zero. So, it only remains to show that
one can really apply Lebesgue's dominated convergence theorem in order to
interchange the limit and the integrals in \eqref{second_term}. For this, let us set
for $\nu\in(-1,1)\setminus\{0\}$ and $\mu\in\R$
\begin{equation*}
L(\nu,\mu):=\int_{\R^d}\underline\d x\,(\F f)(x)
\big\langle\psi,\textstyle\frac1\nu\big\{\eta\big(H_0(\nu x)\big)
\e^{i\frac\mu\nu[H_0(\nu x)-H_0]}
-\eta(H_0)\e^{i\mu x\cdot H'_0}\big\}\psi\big\rangle.
\end{equation*}
By using Lemma \ref{batman} and the fact that $\F f\in\S(\R^d)$, one gets that
$|L(\nu,\mu)|\le{\rm Const.}\;\!(1+|\mu|)$ with a constant independent of $\nu$.
Therefore $|L(\nu,\mu)|$ is bounded uniformly in $\nu\in(-1,1)\setminus\{0\}$ by a
function in $\lone([-1,1],\d\mu)$.

For the case $|\mu|>1$, we first remark that there exists a compact set
$I\subset\R\setminus\kappa(H_0)$ such that $\eta(H_0)=E^{H_0}(I)\eta(H_0)$. Due
to Lemma \ref{Heigen}.(d), there also exists
$\zeta\in C^\infty_{\rm c}\big((0,\infty)\big)$ such that
$$
\eta\big(H_0(\nu x)\big)=\eta\big(H_0(\nu x)\big)\zeta\big(H_0'(\nu x)^2\big)
$$
for all $x\in\R^d$ and $\nu\in\R$. So, using the notations
$$
A_{\nu,\mu}^I(x)
:=\e^{i\frac\mu\nu[H_0(\nu x)-H_0]}E^{H_0}(I)
\equiv\e^{i\frac\mu\nu[H_0(\nu x)-H_0]E^{H_0}(I)}E^{H_0}(I)
$$
and
$$
B_{\mu}^I(x)
:=\e^{i\mu x\cdot H'_0}E^{H_0}(I)
\equiv\e^{i\mu x\cdot H'_0 E^{H_0}(I)}E^{H_0}(I),
$$
one can rewrite $L(\nu,\mu)$ as
$$
L(\nu,\mu)=\int_{\R^d}\underline\d x\,(\F f)(x)
\big\langle\psi,\textstyle\frac1\nu\big\{\eta\big(H_0(\nu x)\big)
\zeta\big(H_0'(\nu x)^2\big) A_{\nu,\mu}^I(x)
-\eta(H_0)\zeta(H_0'^2)B_{\mu}^I(x)\big\}\psi\big\rangle.
$$
Now, using the same technics as in the proof of Lemma \ref{batman}, one shows
that the maps $A_{\nu,\mu}^I:\R^d\to\B(\H_0)$ and $B_{\mu}^I:\R^d\to\B(\H_0)$
are differentiable, with derivatives
$$
\big(\partial_j A_{\nu,\mu}^I\big)(x)
=i\mu(\partial_jH_0)(\nu x)\;\!A_{\nu,\mu}^I(x)
\qquad\hbox{and}\qquad
\big(\partial_jB_{\mu}^I\big)(x)=i\mu(\partial_jH_0)\;\!B_{\mu}^I(x).
$$
Thus, setting
$$
C_j:=(H_0')^{-2}\zeta(H_0'^2)(\partial_jH_0)\;\!\eta(H_0)\in\B(\H_0)
\qquad\hbox{and}\qquad V_x:=\e^{-ix\cdot\Phi},
$$
one can even rewrite $L(\nu,\mu)$ as
$$
L(\nu,\mu)=(i\mu)^{-1}\sum_j\int_{\R^d}\underline\d x\,(\F f)(x)
\big\langle\psi,{\textstyle\frac1\nu}\big\{V_{\nu x}C_j V_{\nu x}^*
\big(\partial_j A_{\nu,\mu}^I\big)(x)-C_j\big(\partial_jB_{\mu}^I\big)(x)\big\}
\psi\big\rangle.
$$
We shall now use repeatedly the following argument: Let $g\in\S(\R^n)$ and let
$X:=(X_1,\ldots,X_n)$ be a family of self-adjoint and mutually commuting operators
in $\H_0$. If all $X_j$ are of class $C^2(\Phi)$, then the operator $g(X)$
belongs to $C^2(\Phi)$, and $\big[[g(X),\Phi_j],\Phi_k\big]\in\B(\H_0)$ for all
$j,k$. Such a statement has been proved in  \cite[Prop.~5.1]{RT10} in a greater
generality. Here, the operator $C_j$ is of the type $g(X)$, since all the operators
$H_0,\partial_jH_0,\ldots,\partial_dH_0$ are of class $C^2(\Phi)$. Thus, we can
perform a first integration by parts (with vanishing boundary contributions) with
respect to $x_j$ to obtain
\begin{align*}
L(\nu,\mu)&=-(i\mu)^{-1}\sum_j\int_{\R^d}\underline\d x\,
\big[\partial_j(\F f)\big](x)\big\langle\psi,{\textstyle\frac1\nu}
\big\{V_{\nu x}C_j V_{\nu x}^*A_{\nu,\mu}^I\big(x)-C_jB_\mu^I(x)\big\}
\psi\big\rangle \\
&\qquad-\mu^{-1}\sum_j\int_{\R^d}\underline\d x\,(\F f)(x)\big\langle\psi,
V_{\nu x}[C_j,\Phi_j]V_{\nu x}^*A_{\nu,\mu}^I(x)\psi\big\rangle.
\end{align*}
Now, the scalar product in the first term can be written as
$$
(i\mu)^{-1}\big\langle\psi,{\textstyle\frac1\nu}\big\{V_{\nu x}DV_{\nu x}^*
\big(\partial_jA_{\nu,\mu}^I\big)(x)-D\big(\partial_jB_{\mu}^I\big)(x)\big\}
\psi\big\rangle
$$
with $D:=(H_0')^{-2}\zeta(H_0'^2)\;\!\eta(H_0)\in\B(\H_0)$. Thus, a further
integration by parts leads to
\begin{align}
L(\nu,\mu)&=-\mu^{-2}\sum_j\int_{\R^d}\underline\d x\,
\big[\partial_j^2(\F f)\big](x)\big\langle\psi,{\textstyle\frac1\nu}
\big\{V_{\nu x}DV_{\nu x}^*A_{\nu,\mu}^I(x)-DB_{\mu}^I(x)\big\}
\psi\big\rangle\label{Hor1}\\
&\qquad-i\mu^{-2}\sum_j\int_{\R^d}\underline\d x\,\big[\partial_j(\F f)\big](x)
\big\langle\psi,V_{\nu x}[D,\Phi_j] V_{\nu x}^*A_{\nu,\mu}^I(x)\psi\big\rangle
\label{Hor2}\\
&\qquad-\mu^{-1}\sum_j\int_{\R^d}\underline\d x\,(\F f)(x)\big\langle\psi,
V_{\nu x}[C_j,\Phi_j]V_{\nu x}^*A_{\nu,\mu}^I(x)\psi\big\rangle\label{Hor3}.
\end{align}
By setting $E_k:=(H_0')^{-4}\zeta(H_0'^2)(\partial_kH_0)\;\!\eta(H_0)\in\B(\H_0)$
and by performing a further integration by parts, one obtains
that \eqref{Hor1} is equal to
\begin{align*}
&i\mu^{-3}\sum_{j,k}\int_{\R^d}\underline\d x\,\big[\partial_j^2(\F f)\big](x)
\big\langle\psi,{\textstyle\frac1\nu}\big\{V_{\nu x}E_k V_{\nu x}^*
\big(\partial_k A_{\nu,\mu}^I\big)(x)-E_k\big(\partial_kB_{\mu}^I\big)(x)\big\}
\psi\big\rangle\\
&=-i\mu^{-3}\sum_{j,k}\int_{\R^d}\underline\d x\,
\big[\partial_k\partial_j^2(\F f)\big](x)\big\langle\psi,{\textstyle\frac1\nu}
\big\{V_{\nu x}E_kV_{\nu x}^*A_{\nu,\mu}^I(x)-E_k B_{\mu}^I(x)\big\}
\psi\big\rangle\\
&\qquad+\mu^{-3}\sum_{j,k}\int_{\R^d}\underline\d x\,\big[\partial_j^2(\F f)\big](x)
\big\langle\psi,V_{\nu x}[E_k,\Phi_k]V_{\nu x}^*A_{\nu,\mu}^I(x)\psi\big\rangle.
\end{align*}
By mimicking the proof of Lemma \ref{batman}, with $\eta(H_0)$ replaced by $E_k$,
one obtains that there exist $\textsc c,s>0$ such that for all $|\mu|>1$,
$x\in\R^d$ and $\nu\in(-1,1)\setminus\{0\}$
$$
\textstyle\big\|\frac1\nu\big\{V_{\nu x}E_kV_{\nu x}^*A_{\nu,\mu}^I(x)
-E_kB_{\mu}^I(x)\big\}\big\|
\le\textsc c\;\!(1+|\mu|)\langle x\rangle^s.
$$
So, the terms \eqref{Hor1} and \eqref{Hor2} can be bounded uniformly in
$\nu\in(-1,1)\setminus\{0\}$ by a function in
$\lone\big(\R\setminus[-1,1],\d\mu\big)$. For the term \eqref{Hor3}, a direct
calculation shows that it can be written as
$$
-i\mu^{-2}\sum_{j,k}\int_\R\underline\d x\,(\F f)(x)
\big\langle V_{\nu x}^*\psi,[C_j,\Phi_j]V_{\nu x}^*C_kV_{\nu x}
\big(\partial_kA_{\nu,-\mu}^{I}\big)(-x)V_{\nu x}^*\psi\big\rangle.
$$
So, doing once more an integration by parts with respect to $x_k$, one also obtains
that this term is bounded uniformly in $\nu\in(-1,1)\setminus\{0\}$ by a function in
$\lone\big(\R\setminus[-1,1],\d\mu\big)$.

The last estimates, together with our previous estimate for $|\mu|\le1$, show that
$|L(\nu,\mu)|$ is bounded uniformly in $|\nu|<1$ by a function in $\lone(\R,\d\mu)$.
So, one can interchange the limit $\nu\searrow0$ and the integration over $\mu$ in
\eqref{second_term}. The interchange of the limit $\nu\searrow0$ and the integration
over $x$ in \eqref{second_term} is justified by the bound obtained in Lemma
\ref{batman}.
\end{proof}

The existence of the usual time delay is now a direct consequence of Theorems
\ref{sym_case} and \ref{equal_sojourn}:

\begin{Theorem}\label{big_one}
Let $H_0$, $f$, $H$ and $\Phi$ satisfy Assumptions \ref{chirimoya}, \ref{commute},
\ref{assumption_f}, \ref{wave} and \ref{chirimoya2}. Let $\varphi\in\H_0^-\cap\D_2$ satisfy
$S\varphi\in\D_2$, \eqref{H-+} and \eqref{grobobo}. Then one has
$$
\lim_{r\to\infty}\tau_r^{\rm in}(\varphi)
=\lim_{r\to\infty}\tau_r(\varphi)
=-\big\langle\varphi,S^*\big[T_f,S\big]\varphi\big\rangle,
$$
with $T_f$ defined by \eqref{nemenveutpas}.
\end{Theorem}

\begin{Remark}\label{rem_com}
In $\ltwo(\R^d)$, the position operators $Q_j$ and the momentum operators $P_j$ are
related to the free Schr\"odinger operator by the commutation formula
$P_j=i\big[-\12\Delta,Q_j\big]$. Therefore, if one interprets the collection
$\{\Phi_1,\ldots,\Phi_d\}$ as a family of position operators, then it is natural (by
analogy to the Schr\"odinger case) to think of
$H_0'\equiv\big(i[H_0,\Phi_1],\ldots,i[H_0,\Phi_d]\big)$ as a velocity operator for
$H_0$. As a consequence, one can interpret the commutation assumption \eqref{grobobo}
as the conservation of (a function of) the velocity operator $H_0'$ by the scattering
process, and the meaning of Theorem \ref{big_one} reduces to the following: If the
scattering process conserves the velocity operator $H_0'$, then the usual and the
symmetrized time delays are equal.

There are several situations where the commutation assumption \eqref{grobobo} is
satisfied. Here we present three of them:
\begin{enumerate}
\item[(i)] Suppose that $H_0$ is of class $C^1(\Phi)$, and assume that there exists
$v\in\R^d\setminus\{0\}$ such that $H_0'=v$. Then the operator $F_f(H_0')$ reduces
to the scalar $F_f(v)$, and $\big[F_f(H_0'),S\big]=0$ in $\B(\H_0)$. This occurs for
instance in the case of Friedrichs-type and Stark operators (see
\cite[Sec.~7.1]{RT10}).

\item[(ii)] Suppose that $\Phi$ has only one component and that $H_0'=H_0$. Then the
operator $F_f(H_0')\equiv F_f(H_0)$ is diagonalizable in the spectral representation
of $H_0$. We also know that $S$ is decomposable in the spectral representation of
$H_0$. Thus \eqref{grobobo} is satisfied for each $\varphi\in\D_0$, since
diagonalizable operators commute with decomposable operators. This occurs in the case
of $\Phi$-homogeneous operators $H_0$ such as the free Schr\"odinger operator (see
\cite[Sec.~7.2]{RT10} and also \cite[Sec.~10 \& 11]{BG91}).

\item[(iii)] More generally, suppose that $F_f(H_0')$ is diagonalizable in the
spectral representation of $H_0$. Then \eqref{grobobo} is once more satisfied for
each $\varphi\in\D_0$, since diagonalizable operators commute with decomposable
operators. For instance, in the case of the Dirac operator and of dispersive systems
with a radial symbol, we have neither $H_0'=v\in\R^d\setminus\{0\}$, nor $H_0'=H_0$.
But if we suppose $f$ radial, then $F_f(H_0')$ is nevertheless diagonalizable in the
spectral representation of $H_0$ (see \cite[Sec.~7.3]{RT10} and
\cite[Rem.~4.9]{Tie09}).
\end{enumerate}
\end{Remark}


\end{document}